\documentclass[final,pdftex]{ectaart}

\RequirePackage[OT1]{fontenc}
\RequirePackage{amsthm}
\RequirePackage[cmex10]{amsmath}
\RequirePackage{natbib}
\RequirePackage[colorlinks,citecolor=blue,urlcolor=blue]{hyperref}

\usepackage{xcolor}
\usepackage{pstricks}
\usepackage{smartdiagram}  
\usepackage[OT1]{fontenc}
\usepackage{natbib}
\usepackage{amsthm}
\usepackage[cmex10]{amsmath}
\usepackage{braket}
\usepackage{stmaryrd}
\usepackage{tikz}
\usetikzlibrary{calc,positioning,shapes.geometric,shapes.symbols,shapes.misc, trees} 
\usepackage{xcolor,colortbl}
\usepackage[caption=false]{subfig}
\usepackage{placeins}
\definecolor{britishracinggreen}{rgb}{0.0, 0.26, 0.15}
\definecolor{cadmiumgreen}{rgb}{0.0, 0.42, 0.24}
\newcommand{\plan}{\pi}
\newcommand{\action}{a}
\newcommand{\brackets}[3]{\left#1#3\right#2}
\newcommand{\tuple}[1]{\brackets{<}{>}{#1}}
\newcommand{\utilfunc}{u}
\newcommand{\state}{s}
\newcommand{\applied}[1]{\leftq #1 \rightq}
\newcommand{\leftq}{\llbracket}
\newcommand{\rightq}{\rrbracket}
\newcommand{\ptask}{\Pi}
\newcommand{\initstate}{\state_{0}}
\newcommand{\init}{{\initstate}}
\newcommand{\range}[1]{[#1]}
\newcommand{\costfunc}{c}

\newcommand{\var}{v}
\newcommand{\eff}{{\mathsf{eff}}}
\newcommand{\pre}{{\mathsf{pre}}}

\newcommand{\valuefunc}{u}
\tikzset{
	start-end/.style={
		draw,
		rectangle,
		rounded corners,
	},
	input/.style={ 
		draw,
		trapezium,
		trapezium left angle=60,
		trapezium right angle=120,
	},
	operation/.style={
		draw,
		rectangle
	},
	loop/.style={ 
		draw,
		chamfered rectangle,
		chamfered rectangle xsep=2cm
	},
	decision/.style={ 
		draw,
		diamond,
		aspect=#1
	},
	decision/.default=1,
	print/.style={ 
		draw,
		tape,
		tape bend top=none
	},
	process rectangle outer width/.initial=0.15cm,
	predefined process/.style={
		rectangle,
		draw,
		append after command={
			\pgfextra{
				\draw
				($(\tikzlastnode.north west)-(0,0.5\pgflinewidth)$)--
				($(\tikzlastnode.north west)-(\pgfkeysvalueof{/tikz/process rectangle outer width},0.5\pgflinewidth)$)--
				($(\tikzlastnode.south west)+(-\pgfkeysvalueof{/tikz/process rectangle outer width},+0.5\pgflinewidth)$)--
				($(\tikzlastnode.south west)+(0,0.5\pgflinewidth)$);
				\draw
				($(\tikzlastnode.north east)-(0,0.5\pgflinewidth)$)--
				($(\tikzlastnode.north east)+(\pgfkeysvalueof{/tikz/process rectangle outer width},-0.5\pgflinewidth)$)--
				($(\tikzlastnode.south east)+(\pgfkeysvalueof{/tikz/process rectangle outer width},0.5\pgflinewidth)$)--
				($(\tikzlastnode.south east)+(0,0.5\pgflinewidth)$);
			}  
		},
		text width=#1,
		align=center
	},
	predefined process/.default=1.75cm,
	man op/.style={ 
		draw,
		trapezium,
		shape border rotate=180,
		text width=2cm,
		align=center,
	},
	extract/.style={
		draw,
		isosceles triangle,
		isosceles triangle apex angle=60,
		shape border rotate=90
	},
	merge/.style={
		draw,
		isosceles triangle,
		isosceles triangle apex angle=60,
		shape border rotate=-90
	},
	action-rational/.style={
		draw,
		circle,
		thick,
		radius=20pt,
	},
	action-stochastic/.style={
		draw,
		circle,
		thick,
		dashed,
		radius=20pt,
	},
	state/.style={
		draw,
		rectangle 
	},
	state/.default=50cm,
	treenode/.style = {shape=rectangle, rounded corners,
		draw, align=center,
		top color=white, bottom color=blue!20},
	root/.style     = {treenode, font=\Large, bottom color=red!30},
	env/.style      = {treenode, font=\ttfamily\normalsize},
	dummy/.style    = {circle,draw}
}

\tikzstyle{state-deterministic} = [draw, rectangle, rounded corners, inner sep=10pt, inner ysep=20pt]
\tikzstyle{state-uncertain} = [draw, rectangle, dashed, rounded corners, inner sep=10pt, inner ysep=20pt]

\tikzstyle{line} = [draw]

\tikzstyle{level 1}=[level distance=1.7cm, sibling distance=2.8cm]
\tikzstyle{level 2}=[level distance=2cm, sibling distance=2.8cm]
\tikzstyle{level 3}=[level distance=2.5cm, sibling distance=2.8cm]
\tikzstyle{level 4}=[level distance=1.7cm, sibling distance=2.8cm]
\tikzstyle{level 5}=[level distance=1.7cm, sibling distance=2.8cm]
\tikzstyle{level 6}=[level distance=2cm, sibling distance=2.8cm]
\tikzstyle{level 7}=[level distance=1.7cm, sibling distance=2.8cm]
\tikzstyle{level 8}=[level distance=2cm, sibling distance=2.8cm]

\tikzstyle{bag} = [text width=4em, text centered]
\tikzstyle{largeBag} = [text centered]
\tikzstyle{end} = [circle, minimum width=3pt,fill, inner sep=0pt]
\tikzstyle{decide} = [draw, diamond, minimum width=.3ex, text width=2em, text centered]
\tikzstyle{myarrow}=[->, thick, shorten >=1pt]
\tikzset{%
	mynode/.style={circle,minimum width=9ex, fill=none,draw}, 
	mynode_stoc/.style={circle, dashed, minimum width=.5ex, fill=none,draw}, 
	myfillnode/.style={circle, fill=black,draw}, 
}

\def\defemb#1#2{\expandafter\def\csname #1\endcsname
	{\relax\ifmmode #2\else\hbox{$#2$}\fi}}

\defemb{cD}{{\cal D}}
\defemb{cP}{{\cal P}}
\defemb{cV}{{\cal V}}
\defemb{cE}{{\cal E}}

\newcommand{\variables}[1]{\cV({#1})}

\startlocaldefs
\numberwithin{equation}{section}
\theoremstyle{plain}
\newtheorem{thm}{Theorem}[section]
\newtheorem{definition}{Definition}[section]

\endlocaldefs

\begin{document}
	
\begin{frontmatter}
\title{Relative Net Utility and the Saint Petersburg Paradox}
\runtitle{Relative Net Utility and the Saint Petersburg Paradox}
 \begin{aug}
 \author{\fnms{Daniel} \snm{Muller}}
 \author{\fnms{Tshilidzi} \snm{Marwala}}
\address{University of Johannesburg, Auckland Park 2000 South Africa            Email: mullerdm@gmail.com; tmarwala@gmail.com}
 \end{aug}
 
\begin{abstract}
	The famous Saint Petersburg Paradox (St. Petersburg Paradox) shows that the theory of expected value does not capture the real-world economics of decision-making problems. Over the years, many economic theories were developed to resolve the paradox and explain gaps in the economic value theory in the evaluation of economic decisions, the subjective utility of the expected outcomes, and risk aversion as observed in the game of the St. Petersburg Paradox. In this paper, we use the concept of the relative net utility to resolve the St. Petersburg Paradox. Because the net utility concept is able to explain both behavioral economics and the St. Petersburg Paradox, it is deemed to be a universal approach to handling utility. This paper shows how the information content of the notion of net utility value allows us to capture a broader context of the impact of a decision's possible achievements. It discusses the necessary conditions that the utility function has to conform to avoid the paradox. Combining these necessary conditions allows us to define the theorem of indifference in the evaluation of economic decisions and to present the role of the relative net utility and net utility polarity in a value rational decision-making process.  
\end{abstract}
\end{frontmatter}
	
\section{Introduction}
The St. Petersburg Paradox is a well-known problem in probability and decision theory. The St. Petersburg Paradox shows that the theory of expected value does not capture the real-world economics of decision-making problems. The problem was formulated by Nicolas Bernoulli in 1713 and was introduced by his cousin, Daniel Bernoulli~\cite{Bernoulli1Daniel}, to the Imperial Academy of Sciences in St. Petersburg as follows:

\begin{quotation}
{\em Peter tosses a coin and continues to do so until it lands on "heads" when it comes to the ground. He agrees to give Paul one ducat if he gets "heads" on the very first throw, two ducats if he gets it on the second, four if on the third, eight if on the fourth, and so on. With each additional throw, the number of ducats he must pay is doubled. Suppose we seek to determine the value of Paul's expectation.}
\end{quotation}

In the paper, ``Exposition of a New Theory on the Measurement of Risk", \cite{Bernoulli1Daniel} suggests a solution to the paradox which is based on the subjectivity of utility of the expected outcome of the game. The St. Petersburg Paradox is referred to by \cite{samuelson1977st} as ''a dramatic and even over-dramatic case", and it illustrates that real-world economics of decision-making problems are not captured by the theory of expected value. In the history of economic literature, St. Petersburg Paradox has been a key factor in the development of the utility function theories among many other ideas that were born from its resolution attempts. Over the years, the question regarding a general utility function that captures arbitrary payoff functions has gained attention in a wide range of scientific and industrial fields and has been addressed by many experts. In his analysis of the financial and economic aspects of St. Petersburg Paradox \cite{libor2011financial} states that by using the analysis and applications of this paradox we may could have avoided some catastrophic financial situations. Analyzing high-tech stocks~\cite{szekely2004st} came to the same conclusion regarding the run-up in stock prices in the late 1990s and the subsequent declines in 2000.

Over the years, most of the solutions and attempts to explain the paradox have dealt with definitions of subjective utility function and manipulations of the utility function to fit the real world observed behavior (\cite{bernoulli19931738,weber1834pulsu,fechner1860elemente,von1947theory,von2007theory}). Bounded utility function, bounded time or bounded number of coin flips assumption remain the most favored explanations of the paradox to date (\cite{menger1934unsicherheitsmoment,samuelson1977st,menger1979role,arrow1970essays,aumann1977st,Seidl2013}). 

In this paper, we take several perspectives on the paradox, including the cost-benefit and time perspective. We start with an analysis of the paradox and present the solution attempts and the theories of expected utility that were developed from the resolution attempts. We define the notion of net utility and define a break-even point as a reference point for net utility inference. We continue with a definition of investment (or game) position and discuss the net utility of changing the position of investment and show how the reference point of break-even point along with the net utility notion allows us to capture the utility of alternative positions in the game. We then define a dynamic reference point and present an incremental net utility evaluation procedure to evaluate the alternative game positions. This allows us to present several new properties in the structure of the problem defined in the St. Petersburg Paradox and to establish the equilibrium position that maximizes the utility while preserving the invested resources. We formulate the Theorem of Indifference concerning the expected utility payoff. This allows us to discuss the opportunity cost and the time factor in investment (and St. Petersburg Paradox game). Incorporating time factor, the objective in the evaluation of the game becomes maximizing the expected net utility while minimizing the time invested in creating a net utility. We treat time as a resource and extend our observations to a general set of resources, establishing resources based criteria for tie-breaking points where there is indifference concerning the expected value of the game. 

Finally, we present the universal theory of net utility and show the role of net utility in a rational decision-making process and how the value of net utility allows for a more informed decision-making process. We present two necessary conditions that the utility function has to conform to avoid the St. Petersburg Paradox. The first condition is the boundness of the utility function. The second condition concerns the net utility polarity, i.e., the net change of the utility value that is achieved due to the decision-making process. The net utility polarity of a decision process must be positive, relative to the starting point of the decision process. The relation point is the break-even position, where the most basic break-even position in any investment is the decision not to buy in, i.e., doing nothing. Combining these necessary conditions allows us to define the theorem of indifference and to present the role of the relative net utility in a value rational decision-making process. By the term {\em value rational}, we refer to a decision-making process with net positive polarity, that is, a total net positive utility change. We conclude with the applications of our observation from St. Petersburg Paradox and net utility in artificial intelligence (AI) and Economic systems and point out some future research directions.

\section{Saint Petersburg Paradox}

We begin by presenting a table summarizing the details of the game in Table~\ref{tbl:paradoExpectedPayoffTable1}. In Table~\ref{tbl:paradoExpectedPayoffTable1} each row shows the number of coin tosses until Paul gets 'heads' for the first time and the game ends. The first column in table~\ref{tbl:paradoExpectedPayoffTable1} represents the number, $k$, of coin tosses. The second column represents the probability of getting 'heads' on $k$ coin toss. The third column represents the reward derived from getting 'heads' on $k$ toss. The fourth column represents the expected payoff from getting 'heads' on $k$ toss, which is calculated by multiplying the probability of getting "heads" on $k$ toss with the respective value of the obtained reward.

 \renewcommand{\arraystretch}{1.6}
\begin{table}[htb!]
	\centering
	\setlength{\tabcolsep}{.18em}
	\begin{tabular}{|c|c|c|c|}
		k (num of throws) & $P(k)=\frac{1}{2^k}$ & Prize & Expected Payoff\\
		\hline
		1 &  $\frac{1}{2}$ &  1 & $\frac{1}{2}$\\ [0.05cm]
		\hline
		2 &  $\frac{1}{4}$ &  2 & $\frac{1}{2}$\\ [0.05cm]
		\hline
		3 &  $\frac{1}{8}$ &  4 & $\frac{1}{2}$\\ [0.05cm]
		\hline
		4 &  $\frac{1}{16}$ &  8 & $\frac{1}{2}$\\ [0.05cm]
		\hline
		5 &  $\frac{1}{32}$ & 16 & $\frac{1}{2}$\\ [0.05cm]
		\hline
		6 &  $\frac{1}{64}$ & 32 & $\frac{1}{2}$\\ [0.05cm]
		\hline  
		7 &  $\frac{1}{128}$ & 64 & $\frac{1}{2}$\\ [0.05cm]
		\hline   
        $\dots$ &  $\dots$ & $\dots$ & $\dots$\\ 
	\end{tabular}
    \caption{\label{tbl:paradoExpectedPayoffTable1} Illustration of the St. Petersburg Paradox game in terms of number of tosses (k), the probability to get 'heads' on toss $k$, the respective prize, and the respective expected payoff value.}
\end{table} 

The mathematical value of a fair price for playing the game is the accumulation of expected payoffs, 
$$\frac{1}{2} + \frac{1}{2} + \frac{1}{2} + \frac{1}{2} .... = \infty.$$
This implies that a player will be willing to pay any price, up to infinity, to participate in the game. However, most of the agents will pay a small amount of money to participate in the game, contradicting the prediction of the expected value theory. This is the paradox. While most of the literature talks about the paradox of the infinite value of the game, there is an additional paradox which has been ignored. The minimum win if one chooses to participate in this game is 1 ducat even though utility theory expects a win of a half a ducat. This is because the definition of expected utility is a first-order approximation of the real utility. This can be viewed as the first-order components of a Taylor expansion of the real utility.

The failure of the expected value theory to predict the real economic value of the game has attracted many researchers over the years and served an important role in the development of economic theories. A survey of the history of ideas, theoretical and empirical, to explain and solve the paradox was conducted by \cite{samuelson1977st,neugebauer2010moral,Seidl2013,Cox2019}. The explanations of this well established challenge of the expected utility theory concentrated on diminishing marginal utility function, treating small probabilities as zero, the boundness of utility function, time boundness or a bounded number of flip coins.

The utility function concept was first introduced by~\cite{Bernoulli1Daniel} to represent the subjective value of winnings and is also referred to as the moral value or the moral hope by Gabrial Cremer in Bernoulli's paper. Gabriel Cremer and Daniel Bernoulli explained the paradox by decreasing marginal utility and suggested concave transformation to solve the paradox. Cremer suggested a square root of the original utility function and Bernoulli suggested a natural logarithm. The subjectivity assumption was experimentally confirmed ~\cite{weber1834pulsu,fechner1860elemente} and formalized in the Expected Utility model~\cite{von1947theory,von2007theory}. Some experimental works showed that the concepts of marginal utility and perception of small probabilities as zero cannot explain the paradox but support Bernoulli's observations of risk-aversion tendency in decision-making~\cite{cox2011empirical,neugebauer2010moral,Cox2019}. Bounded utility function or bounded number of coin flips assumption remains the most favored explanation of the paradox to date.

\cite{menger1934unsicherheitsmoment} work on the St. Petersburg Paradox showed that the marginal diminishing utility transformation is not sufficient since the rewards can increase faster than the rate at which the utility diminishes. This observation led to the conclusion that boundedness of the utility function or coin flips is a necessary condition to prevent the occurrence of a St. Petersburg
Paradox~\cite{menger1934unsicherheitsmoment,menger1979role,arrow1970essays,aumann1977st,Seidl2013}. The time boundedness assumption~\cite{brito1975becker,aumann1977st,cowen1988time} suggests that limited
time will bind the utility function. Even infinite rewards will create finite utility with the natural bound of individual life. In this paper, in the context of boundedness in the St. Petersburg paradox, we discuss the basic bounds of the budget and time of the player and provide a framework to trace and set bounds on the net utility change between alternative game scenarios, i.e., the boundedness of the incremental change in net utility rather than the boundedness of the utility function outcomes. 

Non-Expected Utility theories were developed due to the paradox in the expected utility theory, and these include: the dual theory of choice under risk~\cite{MenachemYaari1987}, prospect theory~\cite{kahneman1979prospect} and the cumulative prospect theory~\cite{tversky1992advances}. These theories do not address the St. Petersburg Paradox, however, they provide several important observations that can be leveraged to solve the paradox. 

The concept of a reference point and relative utility is observed in the prospect theory. Prospect theory assumes the existence of a reference point for the perception of decisions outcomes as losses or gains. Recent work by~\cite{Werner2019} provides a survey on models with reference points and suggests a model to reveal a reference point in the prospect theory. In the next section, we leverage prospect theory observations on the reference point and relative utility to build a framework to reveal dynamic reference points for net utility change of game positions in the St. Petersburg Paradox, which we exploit to choose a policy and evaluate the game. The theoretic contribution of this paper suggests an explanation for the experimental observations in the expected utility theory and behavioral economics on reference points and relative utility as well as risk-aversion and diminishing marginal utility.    

\subsection{Petersburg Game decision-making Process Diagram}
Representation of the game as in Table~\ref{tbl:paradoExpectedPayoffTable1} is common in most discussions about St. Petersburg Paradox, although it is missing a few important details regarding the decision-making process of participating in the game and the evaluation of possible scenarios in the game. In what follows we take a deeper look at the decision-making process in the St. Petersburg Paradox game and provide a diagram to illustrate it.

\begin{figure}[ht!]
\centering
\makebox[0pt]{\footnotesize
\begin{tikzpicture}
        [
        	grow=down,
            emph/.style={edge from parent/.style={dashed, black, draw}},
    		norm/.style={edge from parent/.style={solid,black,thin,draw,->, thick, shorten >=1pt}},
        	sloped
        ]
\node[state,norm] (init) {Init}
child [norm]{
	node (decide) [decide] {Buy In?} 
        child[norm] {
              node (pay) [mynode] {Pay}
              child[norm, grow=right]{
                    node (flip1) [mynode_stoc] {toss-1}
                    child[norm] {
                        node[bag]{$H$}           
                        edge from parent 
                        node[above] {$H$, $\frac{1}{2}$}
                    }
                    child [norm]{
                          node[mynode_stoc] {toss-1}  
                          child [norm]{
                                  node[bag]{$T,H$}   
                                  edge from parent
                                  node[above] {$H$, $\frac{1}{2}$}
                          }
                          child [norm]{
                                node [] {}
                                child [emph]{
                                          node [mynode_stoc] {toss-n}
                                          child [norm]{
                                              node[bag]{$T,T,T, \dots ,H$}  
                                              edge from parent
                                              node[above] {$H$, $\frac{1}{2}$}
                                          }
                                          child [norm]{
                                              child [emph] {
                                                node [bag] {}
                                                edge from parent                                    
                                          	  }
                                          edge from parent         
                                          node[above] {$T$, $\frac{1}{2}$}
                                          }  
                                    edge from parent 
                                }
                                edge from parent
                                node[above] {$T$, $\frac{1}{2}$}
                          }   
                          edge from parent
                          node[above] {$T$, $\frac{1}{2}$}
        			}
        			edge from parent     
                    node[above] {$-C$}
            	}
            edge from parent  
            node (dum) [] {}
            node[above of=dum, node distance=0.3cm] {Yes}
  		}
        node (dummy_no) [right of=decide, node distance=1.5cm] {}
        node (dummy_no_mid) [right of=decide, node distance=1.1cm] {}
        node [above of=dummy_no_mid, node distance=0.4cm] {No}
        node(dummy_above)[right of=init, node distance=1.5cm] {}
};
 \draw[myarrow] (decide.east) -- (dummy_no.center) -- (dummy_above.center) -- (init.east);
    \draw (0,-6.5) -- (8.5,-6.5);
    \draw[dashed] (8.5,-6.5) -- (12,-6.5);
    \draw (2.5,-6.5) node[below=3pt] {$ toss-1 $} node[above=3pt] {$ t=0 $};
    \draw (4.5,-6.5) node[below=3pt] {$ toss-2 $} node[above=3pt] {$ t=1 $};
    \draw (8,-6.5) node[below=3pt] {$ toss-n $} node[above=3pt] {$ t=n $};
    \draw (0,-8.5) -- (8.5,-8.5);
    \draw[dashed] (8.5,-8.5) -- (12,-8.5);
    \draw (0,-8.5) node[below=3pt] {$Init$ $Pos.$} node[above=3pt] {$ B=b $};
    \draw (2.5,-8.5) node[below=3pt] {$ Buy $ $ In$} node[above=3pt] {$ b-c $};
    \draw (4.5,-8.5) node[below=3pt] {$ toss-2 $} node[above=3pt] {$ b-c+1 $};
    \draw (8,-8.5) node[below=3pt] {$ toss-n $} node[above=3pt] {$ b-c+2^{n-1} $};
    \draw (0,-10.5) -- (8.5,-10.5);
    \draw[dashed] (8.5,-10.5) -- (12,-10.5);
    \draw (0,-10.5) node[below=3pt] {$Init$ $Pos.$} node[above=3pt] {$ E=0 $};
    \draw (2.5,-10.5) node[below=3pt] {$ Buy $ $ In$} node[above=3pt] {$ \frac{1}{2} $};
    \draw (4.5,-10.5) node[below=3pt] {$ toss-2 $} node[above=3pt] {$ \frac{1}{2} $};
    \draw (8,-10.5) node[below=3pt] {$ toss-n $} node[above=3pt] {$ \frac{1}{2} $};
    \draw (0,-12.5) -- (8.5,-12.5);
    \draw[dashed] (8.5,-12.5) -- (12,-12.5);
    \draw (0,-12.5) node[below=3pt] {$Init$ $Pos.$} node[above=3pt] {$ E(B)=b $};
    \draw (2.5,-12.5) node[below=3pt] {$ Buy $ $ In$} node[above=3pt] {$ b-c+\frac{1}{2} $};
    \draw (4.5,-12.5) node[below=3pt] {$ toss-2 $} node[above=3pt] {$ b-c+\frac{1}{2} $};
    \draw (8,-12.5) node[below=3pt] {$ toss-n $} node[above=3pt] {$ b-c+\frac{1}{2} $};
\end{tikzpicture}
}
\caption{\label{fig:main_game_diagram} An Illustration of a single game as described in St. Petersburg Paradox. The game illustrated with a decision-event diagram on a scale of time, budget, expected value and expected total budget associated with achievable positions in the game.}
\end{figure}

Figure~\ref{fig:main_game_diagram} illustrates an event-decision diagram of a single St. Petersburg paradox game that Peter suggests to Paul. The diagram represents the game scenarios and the decision process that Paul has to go through in the game that Peter suggests. Generally speaking, the process constitutes the states in the game, one decision, and one deterministic action followed by a sequence of stochastic events. The process starts from an initial state as illustrated with a rectangle and represents the position before the game starts. A decision state is illustrated with a rhombus and represents Paul's decision of whether to participate in the game or not. An action state is illustrated with a circle and represents the action of Paul's payment to Peter in case he decides to participate in the game. A sequence of stochastic events illustrated with dashed circles represents the sequence of tosses until the coin lands on 'heads'.

The game as illustrated in Figure~\ref{fig:main_game_diagram} can be analyzed from different time, value and budget perspectives. In what follows we will cover some of the main perspectives used to analyze the game and to estimate the expected value of the game.
 
We now take a closer look at the different game scenarios in St. Petersburg Paradox which is represented by the diagram in Figure~\ref{fig:main_game_diagram}. Paul's budget at the initial state of the game is $B=b$ ducats. Peter suggests that Paul participate in the game which leaves Paul with two choices: If Paul declines Peter's offer, he stays in his initial position with the initial budget of $b$ ducat. If he accepts Peter's offer, he has to pay $c$ ducats to participate in the game. Once Paul decides to buy into the game, the decision phase is done, the action of paying for the game is applied and the reward is determined according to the outcome of a sequence of stochastic events. By stochastic events, we refer to the sequence of coin tosses until the coin lands on 'heads'. 

Given the terms of the game, Paul can calculate the expected {\bf brake-even point} for the amount of money $c$ that was paid to participate in the game, within his budget restrictions $c<B$, that he is willing and able to pay to participate. Based on the fact that only one scenario of events from the presented scenarios can occur in one game, Paul chooses his expected position from different prospective scenarios, i.e., his payoff expectation for brake-even or an improved position from participation in the game. As illustrated in this example, the expected value of each scenario of sequence of tosses that ends with 'heads' is $\frac{1}{2}$ which can be calculated by the multiplication of scenario probability and the related prize as illustrated in Table~\ref{tbl:paradoExpectedPayoffTable1}. Hence for each sequence of tosses that ends with 'heads' the break-even expected value is $\frac{1}{2}$. This is without taking the cost of time into account. Taking into account the invested time in the game (which can be represented as a monetary value through lost alternatives of investments also referred to as the opportunity cost), we can formulate a decreasing utility function. The rate of the decrease is dependent on the alternative investments. Paul is aware of his evaluation methods of alternative investments (or opportunities).

\section{Net Utility and St. Petersburg Paradox}

\subsection{Evaluation with Respect to Break-Even Point}
To put it simply, the break-even point is where the costs we pay for participating in the game is equal to the benefits that are derived from participating in the game. 
The first and the most basic break-even point in the game is the initial state which is achieved with the decision of not participating in the game. 

\subsection{The Decision-Choice of Not Participating in the Game}
The decision to do nothing, depicted in decision nodes 'Buy In' in the game diagram in Figure~\ref{fig:main_game_diagram}, is always implicitly available. To evaluate the alternatives and the expected outcomes of the game, we set the decision-choice of doing nothing, i.e. not participating in the game, as the first and the most basic reference point.

The reference point of 'doing nothing' will allow us to define the break-even point of the game. When we take into account the costs of participating in the game and the budget of the participant, the reference point of 'doing nothing' is a break-even point. We define the break-even point as follow.

\begin{definition}
Given a sequence of events $e_0,e_1,e_2 \dots e_k$ a cost function $\costfunc({e_i})$ and a value function $\valuefunc({e_i})$ for each $i \leq k$ the break even point of the game $BEP$ occurs when $\sum_{i=1}^{k}{\valuefunc({e_i})} = \sum_{i=1}^{k}{\costfunc({e_i})}$. 
\end{definition}

We define the break-even as the utility achieved in the last event (heads on a coin toss) of a sequence of events. The cost, however, is paid at the initial point of the sequence of events with the decision to participate in the game. The expected payoff depends on both "tails" and "heads". The number of "tails" determines the payoff whereas the "heads" is a critical indicator of when the game is stopped and the payoff effected.

 \renewcommand{\arraystretch}{1.6}
\begin{table}[htb]
	\centering
	\setlength{\tabcolsep}{.18em}
	\begin{tabular}{|c|c|c|c|}
		k (num of throws) & $P(k)=\frac{1}{2^k}$ & Prize & Expected Payoff\\
		\hline
\rowcolor{gray}	Init/Fold &  -- & 0 & 0\\ [0.05cm]
		\hline
		1 &  $\frac{1}{2}$ &  1 & $\frac{1}{2}$\\ [0.05cm]
		\hline
		2 &  $\frac{1}{4}$ &  2 & $\frac{1}{2}$\\ [0.05cm]
		\hline
		3 &  $\frac{1}{8}$ &  4 & $\frac{1}{2}$\\ [0.05cm]
		\hline
		4 &  $\frac{1}{16}$ &  8 & $\frac{1}{2}$\\ [0.05cm]
		\hline
		5 &  $\frac{1}{32}$ & 16 & $\frac{1}{2}$\\ [0.05cm]
		\hline
		6 &  $\frac{1}{64}$ & 32 & $\frac{1}{2}$\\ [0.05cm]
		\hline  
		7 &  $\frac{1}{128}$ & 64 & $\frac{1}{2}$\\ [0.05cm]
		\hline   
        $\dots$ &  $\dots$ & $\dots$ & $\dots$\\ 
	\end{tabular}
    \caption{\label{tbl:paradoExpectedPayoffTable2} Illustration of the St. Petersburg Paradox game as in Table~\ref{tbl:paradoExpectedPayoffTable1} with the initial position description added, colored in gray, which describes additionally the 'init' state and the decision of not participating the game.}
\end{table} 

 
Table~\ref{tbl:paradoExpectedPayoffTable2} extends Table~\ref{tbl:paradoExpectedPayoffTable1} with the details of the initial position at the initial state, which is depicted in nodes 'init' in the game diagram in Figure~\ref{fig:main_game_diagram}. The additional row 'Init/Fold' in Table~\ref{tbl:paradoExpectedPayoffTable2} represents the pre-decision state where there is no reward and no costs paid to participate in the game. It also represents the decision "Fold", that is the decision not to participate in the game in which case the reward, as well as the expected payoff, are 0. 

\subsection{Net Utility of Changing Position}

In what follows we present an incremental evaluation approach to find the optimal position in the Saint Petersburg Paradox game.

Using the  Oxford dictionary definition~\cite{dictionary2008oxford} we define a Game Position as follows:
\begin{definition}
A Game Position is:\\
\begin{itemize}
    \item A person's point of view or attitude towards the game payoff
    \item The extent to which an investor, dealer, or speculator has made a commitment in the market by buying or selling securities.
\end{itemize} 
\end{definition}

\begin{definition}
Given a set of game positions $p_0,p_1...p_k...$, the utility of game position $p_k$ in the St. Petersburg Paradox, $\utilfunc(p_k)$ is the expected payoff of a game with $k$ coin flips.\\
\end{definition}

\renewcommand{\arraystretch}{1.6}
\begin{table}[htb]
	\centering
	\setlength{\tabcolsep}{.18em}
    \newcolumntype{g}{>{\columncolor{gray}}c}
	\begin{tabular}{|c|c|c|c|g|}
		k (num of throws) & $P(k)=\frac{1}{2^k}$ & Prize & Expected Payoff & Net Utility\\
		\hline
\rowcolor{gray}	Fold &  -- & -- & 0 & 0\\ [0.05cm]
		\hline
		1 &  $\frac{1}{2}$ &  1 & $\frac{1}{2}$ & ${\bf\frac{1}{2}}$\\ [0.05cm]
		\hline
		2 &  $\frac{1}{4}$ &  2 & $\frac{1}{2}$ &$\frac{1}{2}$\\ [0.05cm]
		\hline
		3 &  $\frac{1}{8}$ &  4 & $\frac{1}{2}$ &$\frac{1}{2}$\\ [0.05cm]
		\hline
		4 &  $\frac{1}{16}$ &  8 & $\frac{1}{2}$ &$\frac{1}{2}$\\ [0.05cm]
		\hline
		5 &  $\frac{1}{32}$ & 16 & $\frac{1}{2}$ &$\frac{1}{2}$\\ [0.05cm]
		\hline
		6 &  $\frac{1}{64}$ & 32 & $\frac{1}{2}$ &$\frac{1}{2}$\\ [0.05cm]
		\hline  
		7 &  $\frac{1}{128}$ & 64 & $\frac{1}{2}$ &$\frac{1}{2}$\\ [0.05cm]
		\hline   
        $\dots$ &  $\dots$ & $\dots$ & $\dots$ &$\dots$\\ [0.05cm]
	\end{tabular}
    \caption{\label{tbl:paradoExpectedPayoffTable3} Illustration of the St. Petersburg Paradox game with the initial position details (colored in gray row) and a "Net Utility" with respect to initial position (colored in gray column).}
\end{table}  

Our initial reference point is the break-even point in the initial state with the value of 0 (no costs nor benefit produced in the game). Having an initial reference point, we can evaluate other applicable positions of Paul in the game with respect to the initial break-even point. 

We define the net utility of game position with respect to a reference point point as follows. 
\begin{definition} 
Given a set of game positions $p_0,p_1...p_k...$, the net utility of game position $p_k$ with respect to a reference position $p_i$ is the net change in the expected payoff, i.e., $\utilfunc(p_k) - \utilfunc(p_i)$.
\end{definition}

Assuming the initial position with $\utilfunc(p_0) = 0$ we have a special case where the utility of game position $p_k$ is also the net change in the expected payoff from the initial position where $i=0$ and equals, $\utilfunc(p_k)$. However, in most of the real-world scenarios, we cannot assume that zero utility for the state in which we consider an investment or gamble decision. Hence, we can not assume that considering the utility of an arbitrary game position will reflect the real impact of obtaining it. In what follows, we show how the informativeness of the notion of net utility allows us to capture a broader context of the impact of a decision's possible achievements.

Table~\ref{tbl:paradoExpectedPayoffTable3} extends Table~\ref{tbl:paradoExpectedPayoffTable2} with the details of the net utility of changing position with respect to the initial state. As Table~\ref{tbl:paradoExpectedPayoffTable3} shows, in terms of the net utility with respect to our initial position, there is indifference with regard to each position with $k>0$.

\subsection{Dynamic Reference Point for Net Utility of Changing Position}

Having an initial reference point, we can evaluate other applicable positions of Paul in the game and update the reference point dynamically when a reference point with better utility is found.

	\renewcommand{\arraystretch}{1.6}
	\begin{table}[htb]
		\centering
		\setlength{\tabcolsep}{.18em}
        \newcolumntype{g}{>{\columncolor{gray}}c}
		\begin{tabular}{|c|c|c|c|g|g|}
			k (num of throws) & $P(k)=\frac{1}{2^k}$ & Prize & Expected Payoff & Net Utility & Relative Net Utility\\
			\hline
\rowcolor{gray}	Fold &  -- & -- & 0 & 0 & 0\\ [0.05cm]
			\hline
			\hline
			\rowcolor{yellow}1 &  $\frac{1}{2}$ &  1 & $\frac{1}{2}$& $\frac{1}{2}$ & ${\bf\frac{1}{2}}$\\ [0.05cm]
			\hline
			\hline
			2 &  $\frac{1}{4}$ &  2 & $\frac{1}{2}$ & $\frac{1}{2}$ &0\\ [0.05cm]
			\hline
			3 &  $\frac{1}{8}$ &  4 & $\frac{1}{2}$ & $\frac{1}{2}$ &0\\ [0.05cm]
			\hline
			4 &  $\frac{1}{16}$ &  8 & $\frac{1}{2}$ & $\frac{1}{2}$ &0\\ [0.05cm]
			\hline
			5 &  $\frac{1}{32}$ & 16 & $\frac{1}{2}$ & $\frac{1}{2}$ &0\\ [0.05cm]
			\hline
			6 &  $\frac{1}{64}$ & 32 & $\frac{1}{2}$ & $\frac{1}{2}$ &0\\ [0.05cm]
			\hline  
			7 &  $\frac{1}{128}$ & 64 & $\frac{1}{2}$ & $\frac{1}{2}$ &0\\ [0.05cm]
 			\hline   
            $\dots$ &  $\dots$ & $\dots$ & $\dots$ & $\dots$ & $\dots$\\ [0.05cm]

		\end{tabular}
        \caption{\label{tbl:paradoExpectedPayoffTable4} Illustration of the St. Petersburg Paradox game with the initial position details (colored in gray row) and a "Relative (dynamic) Net Utility" with respect to best position found so far (colored in gray column). Highlighted in yellow is the first (and last) beneficial changing position decision}
	\end{table}  


\begin{figure}[htb]
	\footnotesize
    	\begin{tabbing}
		\underline{Incremental-Net-Utility-Evaluation}\\ 
		\; $RefPos := 0$ \;\;\;\;\;\;\;\;\textcolor{blue}{\footnotesize \em // Initialize Reference Position}\\
		\; $CurrPos := 0$
		\;\;\;\;\;\;\textcolor{blue}{\footnotesize \em // Initialize Current Position}\\
		\; $NUC := 0$ 
		\;\;\;\;\;\;\;\;\;\;\;\textcolor{blue}{\footnotesize \em // Initialize Net Utility Change}\\
		\; $NUC* := 0$ 
		\;\;\;\;\;\;\;\;\;\textcolor{blue}{\footnotesize \em // Initialize Best Net Utility Change so far}\\
		\; {\bf loop:} \\
		\;\;\;\; $CurrPos = CurrPos + 1$\\ 
		\;\;\;\; $NUC = EU(CurrPos) - EU(RefPos)$
		\;\;\;\;\;\;\textcolor{blue}{\footnotesize \em // Calculate Expected Net Utility Change}\\
		\;\;\;\; {\bf if} $NUC >  NUC*$:\\
		\;\;\;\;\;\;\;\; $NUC* =  NUC$\\
		\;\;\;\;\;\;\;\; $RefPos = CurrPos$\\
		\;\;\;\; {\bf if} Termination-criteria:\\
		\;\;\;\;\;\;\;\; return $NUC*$, $RefPos$\\

	\end{tabbing}
	
	\caption{\label{fig:inc-net-util} Incremental net utility evaluation with dynamic reference point update}
\end{figure}

Table~\ref{tbl:paradoExpectedPayoffTable4} extends Table~\ref{tbl:paradoExpectedPayoffTable3} with the details of the net utility of changing position with respect to the best position found so far. Figure~\ref{fig:inc-net-util} depicts a pseudo-code of an incremental net utility evaluation with dynamic reference point which puts into practice our observations on the net utility change in the St. Petersburg Paradox.

In the next section, we present the Theorem of Indifference and resource preserving tie-breaking decision criteria which allows us to define the termination criteria for the incremental net utility evaluation procedure. The Theorem of Indifference will allow us to find the equilibrium in {\em effective}. The net utility change and the resource preserving tie-breaking decision criteria will allow us to define an {\em efficient} evaluation termination within the equilibrium zone.

\FloatBarrier

\section{Theorem of Indifference and Resource Preserving Tie-Breaking Decision Criteria}

In the St. Petersburg game, we pursue an optimal position to improve the net utility. Table~\ref{tbl:paradoExpectedPayoffTable4} shows that the increase in the number of coin flips (the incremental change in game position) after the initial coin flip, results in no change in the expected utility, with respect to the reference point. Although the Payoff might be increasing with the increased number of coin flips, the probability of achieving the payoff is decreasing at a rate that preserves the same expected utility for each game position. This results in the increase of risk in the game, respective to the number of flip coins (i.e. the respective game positions) without the justification of improved expected utility. We now examine how to find a game position that achieves equilibrium in the game which preserves resources without damaging the potential expected utility. Note that in what follows we optimize the evaluation process of choosing Paul's game position and by doing that we find the resource preserving position among equally attractive positions in terms of the expected utility\footnote{This can be seen as a process of planning a solution, i.e. planning a decision in the case of the St. Petersburg Paradox}.

The evaluation process of choosing a position in the game that is described in Figure~\ref{fig:inc-net-util} consists of the evaluation actions of incremental change and evaluation of game positions. Treating the evaluation process as a sequential action application allows us to adopt action planning methods.
Theorem~\ref{thm:suffix} was suggested by \cite{muller2018value} to solve problems of sequential action planning in which the objective is to choose a sequence of actions to maximize agent utility under budget restrictions. It allows us to optimize a sequence of actions that lead to effective net positive change in the utility by truncating actions that exploit resources (such as time) and do not lead to the net positive utility value change. This will allow us to define efficient termination criteria for the incremental net utility evaluation procedure (Figure~\ref{fig:inc-net-util}).  
Put simply, the Theorem states that; for each sequence of actions $\plan$, there is a sub-sequence $\plan'$ that; (i) ends with a net positive utility value action (net utility of the action of changing game position in the St. Petersburg Paradox), (ii) is at most as costly as $\plan$, and (iii) is at least as valuable as $\plan$.

\begin{thm}[Theorem of Indifference]
	\label{thm:suffix}
	Given a game $\ptask$ with an additive utility function $\utilfunc$, for any policy $\plan$ for $\ptask$ such that  $\utilfunc(\state\applied{\plan})>\utilfunc(\init)$, there is a prefix $\plan'$ of $\plan$ such that:
	\begin{enumerate}
		\item $\utilfunc(\init\applied{\plan})\le \utilfunc(\init\applied{\plan'})$, and
		\item for the last action $\action_{last}$ along $\plan'$, we have $\utilfunc(\action_{last})>0$.
	\end{enumerate}
\end{thm}

\begin{proof}
	The proof is by induction on the plan length $n$. 
	For $n=1$, we have $\plan = \tuple{\action_{1}}$, and since $\utilfunc(\state\applied{\plan})>\utilfunc(\init)$, the action $\action_1$ has a positive net utility value. Hence, $\plan' = \plan$ satisfies the claims. Assuming that the claim holds for $n \geq 1$, we now prove it for $n+1$. 
	
	Considering a plan $\plan = \tuple{\action_1,\dots,\action_{n+1}}$, for $i \in \range{n+1}$, let $\plan_i$ denote the prefix of $\plan$ consisting of its first $i$ actions. If the last action $\action_{n+1}$ along $\plan$ has a positive net utility value, then we are done with $\plan' = \plan$. Otherwise, if $\action_{n+1}$ has either negative or zero net utility value, then $\utilfunc(\state\applied{\plan})>\utilfunc(\init)$ in particular implies 
	$\utilfunc(\state\applied{\plan_n}) >\utilfunc(\init)$. If $\plan'$ is a prefix of $\plan_n$ that satisfies the lemma by our assumption of induction, then $\plan'$ also satisfies the lemma with respect to $\plan$ since 
	$\utilfunc(\plan') \geq  \utilfunc(\state\applied{\plan_n}) \geq \utilfunc(\state\applied{\plan}).$ 
\end{proof}

\subsection{Tie-Breaking Decision Criteria}

Although the expected utility is the same for all alternatives when choosing positions in the game, it would be wrong to state that a rational player would play infinite games. In terms of the expected value, there is no benefit to changing position in the game, i.e., the net utility from changing position from one toss, is zero. At the same time, the implicit costs can be enormous in terms of time. These costs in terms of time mean that changing position from one toss to any other position is not optimal. In summary, by logical principles, a rational choice among equal opportunities will optimize the cost of time invested in each opportunity.

\subsection{Utility Tie Braking with Time Factor}
As Table~\ref{tbl:paradoExpectedPayoffTable3} shows, there is indifference in the dimension of the benefit among the alternatives that Paul can choose from. In such a case, Paul can choose a random position or find a tie-breaking criterion that is measured in different dimensions, for instance, the dimension of time. Given two alternatives of optional St. Petersburg games, game $A$ and $B$, with equal expected utility, $u(A)=u(B)$, a rational player will choose the game that optimizes the factor of time. 

We assume that each toss has a time duration and we annotate the duration of a single toss with $\epsilon$. Given a St. Petersburg game $A$ that constitutes of $k$ tosses, with $k>0$, the duration of game $A$ is $k\epsilon$. In a St. Petersburg game $B$ that consists of $i$ tosses where $i<k$, the duration of game $B$ is $i\epsilon$ and $i\epsilon < k\epsilon$. Suppose that in terms of utility there is indifference between game $A$ and $B$, i.e., $u(A)=u(B)$. In terms of time, it is straight-forward that game $B$ is preferable.  Since in terms of the cost of time, game $B$ is preferable to game $B$, a rational player that both maximizes the utility and optimizes the time will choose game $B$ which is time preserving.

We now extend the tie breaking point to a set of game alternatives. Let a set $$STP=X_1,X_2,X_3....X_k$$ be a set of St. Petersburg games with respective expected utilities of $$u(X_1),u(X_2),u(X_3)...u(X_k)$$ and respective time duration $$t(X_1),t(X_2),t(X_3)....t(X_k).$$ Let the expected utility of each two games in the set be equal, i.e for each $X_i,X_j \in STP$ holds $u(X_i)=u(X_j)$. Suppose $$t(X_1)<t(X_2)<t(X_3)....<t(X_k)$$ a rational player will pursue the game $X_1$ which optimizes the time required to achieve the expected utility. In other words, among equally attractive games in terms of the expected utility, a rational player will pursue the time-preserving alternative which is the least number of tosses. 

\subsection{From Time to Resource-Based Utility Tie Breaking Point}
 The tie-breaking point that is based on the time factor can be easily extended to a general case of resource-consuming actions or games. We treat time as a resource when we define the tie-breaking point and pursue the resource preserving scenario among scenarios with equal net utility change. The logic of preserving time or any other resource is the same. 

\section{The Net Utility in (Bounded) Rational decision-making Process}
In this section, we define and discuss the necessary conditions for a value rational decision-making process. By the term {\em value rational}, we refer to a decision-making process with net positive utility polarity, that is, a total net positive utility change. We start with a review of the rational and bounded rational decision-making process. We then present two necessary conditions that the utility function has to conform to avoid the St. Petersburg Paradox. The first condition is the boundness of the utility function. The second condition concerns the net utility polarity, i.e., the net change of the utility value that is achieved as a result of the decision-making process outcome. The net utility polarity of a decision-making process must be positive relative to the starting position. The relation point is the break-even position, where the most basic break-even position in any investment is the decision not to "buy in", i.e., doing nothing. Combining these necessary conditions allows us to define the theorem of indifference (\ref{thm:suffix}) and to present the role of the relative net utility in a rational, sequential rational decision-making process. 
\subsection{Value Rational and Bounded Rational decision-making Process}
Problem solving, planning, and decision-making processes constitute ordered decision choices and decision actions~\cite{kriger1992organizational,marwala2015causality}. Instantaneous decision-choices are the atomic units that compose a decision's action. An optimized decision results in a global maximization of utility \cite{degroot2005optimal,berger2013statistical}. \textbf{\textit{Decision actions}} can be \textbf{\textit{rational}} or \textbf{\textit{irrational}}. Following the work of~\cite{marwala2015causality}, we define the decision action as rational if it results in a globally optimum outcome, based on logical principles and derived from complete and relevant information. A decision based on irrelevant or incomplete information is irrational, and cannot lead to a global optimum solution. A rational decision-making process comprises of rational decision actions, and the entire process is optimized in time and results in a global utility optimum~\cite{grune2012paradoxes,marwala2014artificial,marwala2015causality}. The quantification of the extent of this irrationality is what is called bounded rationality.

The concept of\textbf{\textit{ bounded rationality}}~\cite{simon1957models,simon1990mechanism,simon1991bounded} explains irrational decision-making based on incomplete data analysis which leads to incomplete information for making a decision. The concept of \textbf{\textit{Flexible-bounded rationality}} is based on the fact that advancements in AI and machine data analysis make the bounds of rationality flexible (\cite{marwala2013flexibly}). 

The rational decision-making process constitutes a sequential (inter-related) application of rational decision actions. \textbf{\textit{Rationality is not dividable}}~\cite{marwala2013flexibly}, a process of applying a sequence of depended decision actions cannot be partly rational and partly irrational. Even one irrational action makes the entire decision-making process irrational and compromises the attainment of a globally optimal solution.  The imperfection of information in a single decision action leads an entire decision-making process to be irrational with a sub-optimal end, in time or utility. Advanced information analysis methods can assist in repairing (partially) the imperfection of information to lead to a (partial) rationalization of a decision-making process. 

We now present two conditions that the utility function must conform to in order to allow for a value rational decision-making process.

\subsection{The Boundness of Utility Function}
\label{subsec:utilityBoundness}
~\cite{bassett1987st,samuelson1977st} provides a detailed historical review on the boundness of utility conditions. The first researcher to observe the necessity of bounded utility function to avoid the paradox was \cite{menger1934unsicherheitsmoment}. Menger's observations were later extended and proven by \cite{arrow1970essays}.
The paradox was extensively investigated in the work of \cite{menger1934unsicherheitsmoment}\footnote{the paper translated to English in \cite{menger1979role}}, which among other things devised the Menger theorem, which can be described as follows.

\begin{thm}
A path has a finite price with the use of utility function $f(x)$ for any sequence of possible payoffs $\tuple{\action_{1},\dots,\action_{n}.\dots}$ when and only when function $f(x)$ is limited.
\end{thm}
\cite{arrow1970essays} proves the utility boundness theorem.
\begin{thm}[Utility Boundness Theorem]
\label{thm:expectedUtilBound}
Any utility function which satisfies the conditions of the Expected Utility Theorem must be bounded from above and below.
\end{thm}

The Expected Utility Theorem \cite{arrow1970essays,samuelson1977st,bassett1987st} provides a condition for preference among two actions. 
\begin{thm}[Expected Utility Theorem]
\label{thm:expectedUtil}
It is possible to define a real-valued utility function over actions [or decisions with their implied respective possible outcomes, each with specifiable probabilities], with the following properties: 
\begin{enumerate}
\item $\action_{1} > -\action_{2}$ [i.e., action 1 definitely preferred to action 2] if and only if $\utilfunc(\action_1) > \utilfunc(\action_2)$
\item $\utilfunc(\action) = E[\utilfunc(\action\applied{\state})]$ [i.e., equals the expected
 value of utilities of {\bf outcomes} implied by action a].
\end{enumerate}
\end{thm}

Assuming state $\state_{1}$ to be the original state, the expected utility of action $action$ is the expected utility value of the target state $\state_{2}$, i.e., the utility of the state that is achieved after the application of an action $\action$. In simple terms, $\action_{1}$ is preferable to $\action_{2}$ if the expected value of utilities of outcomes of the application of $\action_{1}$ in state $\state$, $E[\utilfunc(\action_{1}\applied{\state})]$, are greater than the application of $\action_{2}$ in state $\state$, $E[\utilfunc(\action_{2}\applied{\state})]$.

Note that the expected utility theorem takes into account of the utility of the target state, $\utilfunc(\state_2)$, but not the utility of the original state $\utilfunc(\state_1)$. The theorem captures a process of strictly positive net utility action (or decision) application. That is, each action applied necessarily improves the state in which it is applied. However, most real-world decision-making processes, and in particular, sequential decision-making processes, involve the application of actions or sequences of actions of negative utility polarity. Such actions may worsen/decrease the utility of the state in which it is applied \footnote{the discussion is about negative utility, not the costs of actions that can be seen as investments}. In such a case, the theorem provides a recommendation of preferring actions over other actions, although they both might be negative in total net terms. 

Let us consider the following example to stress this point. Consider an original state $\state_{1}$ with two actions $\action_1$ and $\action_2$ applicable in state $\state_{1}$;
\begin{enumerate}
    \item $\utilfunc(\action_1)$ = $E[\utilfunc(\action_{1}\applied{\state_1})]$ = $E[\utilfunc(\state_{2})]$ = 100 
    \item $\utilfunc(\action_2)$ = $E[\utilfunc(\action_{2}\applied{\state_1})]$ = $E[\utilfunc(\state_{3})]$ = 150 
\end{enumerate}
The expected value of utility outcomes of applying action $\action_2$ are greater than those of applying $\action_1$. Using the theorem of expected utility, we will choose to apply action $\action_2$ to achieve a greater utility outcome. It should be noted, to that end, that a piece of available information which is essential for the decision-making process is neglected. The utility of the original state is not incorporated in the preference system that is implied by the theorem of expected utility.

Now let us consider the utility of the original state $\state_1$ and incorporate it into our preferences system.
We are interested in the net utility polarity in our example and consider several different utility values for the original state $\state_1$ as follows:
\begin{enumerate}
    \item $\utilfunc(\state_1)$ = 75
    \begin{enumerate}
        \item $E[\utilfunc(\action_{1}\applied{\state_1})] = 100 > \utilfunc(\state_1)$
        \item $E[\utilfunc(\action_{2}\applied{\state_1})] = 150  > \utilfunc(\state_1)$ 
    \end{enumerate}
    \item $\utilfunc(\state_1)$ = 100
        \begin{enumerate}
        \item $E[\utilfunc(\action_{1}\applied{\state_1})] = 100 = \utilfunc(\state_1)$
        \item $E[\utilfunc(\action_{2}\applied{\state_1})] = 150 > \utilfunc(\state_1)$ 
    \end{enumerate}
    \item $\utilfunc(\state_1)$ = 125
        \begin{enumerate}
        \item $E[\utilfunc(\action_{1}\applied{\state_1})] = 100 < \utilfunc(\state_1)$ 
        \item $E[\utilfunc(\action_{2}\applied{\state_1})] = 150 > \utilfunc(\state_1)$
    \end{enumerate}
    \item $\utilfunc(\state_1)$ = 150
        \begin{enumerate}
        \item $E[\utilfunc(\action_{1}\applied{\state_1})] = 100 < \utilfunc(\state_1)$
        \item $E[\utilfunc(\action_{2}\applied{\state_1})]= 150 = \utilfunc(\state_1)$
    \end{enumerate}
    \item $\utilfunc(\state_1)$ = 175
        \begin{enumerate}
        \item $E[\utilfunc(\action_{1}\applied{\state_1})] = 100 < \utilfunc(\state_1)$
        \item $E[\utilfunc(\action_{2}\applied{\state_1})] = 150 < \utilfunc(\state_1)$
    \end{enumerate}
\end{enumerate}

When the focus is on the outcome of actions, we can compare action $\action_1$ to $\action_2$ in terms of their outcomes, but we miss the impact of the action, which is easily captured using the net utility terms. It is straightforward that $\action_2$ is preferable to $\action_1$. However, are we asking the right question? Our simple example shows that although the outcome of $\action_2$ is greater than the outcome of $\action_1$, in relative net change terms, it does not imply that we should apply any of them since keeping the current position and doing nothing, results in a better outcome than doing something. The theorem of indifference (\ref{thm:suffix}) captures exactly that phenomena in sequential scenarios as a decision-making process, sequential action application, and sequence of events. The theorem exploits the polarity of the net change in sequential scenarios as the bound for the utility function to avoid the St. Petersburg Paradox. In the next section, we take a closer look at the net utility polarity role in the rational decision-making process and reformulate the theorem of expected utility to the theorem on expected net utility.

\subsection{Net Utility Polarity in Value Rational decision-making Process}
\label{subsec:netUtilityPolarity}
Value rational decision-making process (or game/investment) has net positive utility polarity, i.e., it leads to total net positive change. In particular, the last decision (or position/event) of any value rational decision-making process has to be of net positive polarity.

Based on Theorem~\ref{thm:expectedUtil}, we define the expected net utility theorem as follows: The expected net utility theorem is defined in terms of {\bf total change of utility} rather than the {\bf outcomes} solely as in Theorem~\ref{thm:expectedUtil}. In practice, we define the utility of action as the difference between the action's precondition state and the outcomes of the action, rather than the action's outcomes solely. Evaluation of utility function in terms of expected net utility theorem provides a more informative picture of the real impact of the utility function of actions/decisions. The real impact of an action is not fully captured in the outcomes of actions and can vary concerning the utility of the state in which the action is applied.
\begin{thm}[The Expected Net Utility Theorem]
Given a state $\state_{pre}$ and an action $\action$ applicable in state $\state_{pre}$, $\state_{eff}$ is the state that is obtained with an application of action $\action$. It is possible to define a real-valued utility function over actions [or decisions with their implied respective possible outcomes, each with specifiable probabilities], with the following properties: 
\begin{enumerate}
\item $\action_{1} > -\action_{2}$ [i.e., action 1 definitely preferred to action 2] if and only if $\utilfunc(\action_1) > \utilfunc(\action_2)$
\item $\utilfunc(\action) = E[\utilfunc(\action\applied{\state}) - \utilfunc(s)]$ [i.e., equals the expected
 value of {\bf the total net change} of utilities implied by action $\action$].
\item for an action $\action$, the {\bf net utility} of $\action$ is $\valuefunc(\action) = \sum_{\var\in\variables{\eff(\action)}}[{\valuefunc(\eff(\action)[\var])}-{\valuefunc(\pre(\action)[\var])}].$
\end{enumerate}
\end{thm}


The Expected Net Utility Theorem is more informative since it exploits the utility information compactly and allows us effectively (and efficiently) to put this information into practice, i.e., to conduct a more informative preference system. A preference system that based on the Expected Net Utility Theorem takes into account the net utility polarity of actions, allows us to compare the outcomes actions as well as the total net change achieved by the actions. The following example shows the informativeness of a preference system based on net utility. In some cases, actions with positive outcomes might bring negative utility in net terms.

\begin{enumerate}
    \item $\utilfunc(\state_1)$ = 75
    \begin{enumerate}
        \item $ E[\utilfunc(\action_{1}\applied{\state_1}) - \utilfunc(\state_1)] = 25$
        \item $E[\utilfunc(\action_{2}\applied{\state_1})- \utilfunc(\state_1)] = 75$
    \end{enumerate}
    \item $\utilfunc(\state_1)$ = 100
        \begin{enumerate}
        \item $E[\utilfunc(\action_{1}\applied{\state_1})- \utilfunc(\state_1)] = 0$
        \item $E[\utilfunc(\action_{2}\applied{\state_1})- \utilfunc(\state_1)] = 50$
    \end{enumerate}
    \item $\utilfunc(\state_1)$ = 125
        \begin{enumerate}
        \item $E[\utilfunc(\action_{1}\applied{\state_1})- \utilfunc(\state_1)] = -25$
        \item $E[\utilfunc(\action_{2}\applied{\state_1})- \utilfunc(\state_1)] = 25$
    \end{enumerate}
    \item $\utilfunc(\state_1)$ = 150
        \begin{enumerate}
        \item $E[\utilfunc(\action_{1}\applied{\state_1})- \utilfunc(\state_1)] = -50$
        \item $E[\utilfunc(\action_{2}\applied{\state_1})- \utilfunc(\state_1)] = 0$
    \end{enumerate}
    \item $\utilfunc(\state_1)$ = 175
        \begin{enumerate}
        \item $E[\utilfunc(\action_{1}\applied{\state_1})- \utilfunc(\state_1)] = -75$
        \item $E[\utilfunc(\action_{2}\applied{\state_1})- \utilfunc(\state_1)] = -25$
    \end{enumerate}
\end{enumerate}

\subsection{The Theorem of Indifference in Value Rational decision-making Process}
Combining the boundness (Section \ref{subsec:netUtilityPolarity}) and the net utility polarity (Section \ref{subsec:utilityBoundness}) conditions for utility function, we formulated the theorem of indifference (Theorem \ref{thm:suffix}) that recognizes the suffix of a sequence of events, decisions, or changes of position that do not bring positive change (that is, they are not of net positive polarity). The theorem of indifference combines the two conditions over the utility function in a sequential decision-making process in a sequential game/investment. It puts it into practice in the algorithm of incremental utility improvement evaluation provided in Figure~\ref{fig:inc-net-util}.

\section{The Universal Concept of Net Utility}
In 1901 \cite{MaxPlanck} published a paper that explained the black body radiation problem by assuming that energy is in the form of little packets called quanta. In 1905 \cite{AlbertEinsten} published a paper that explained the photoelectric effect using the concept of quanta that was proposed by Max Planck. Because the theory of quanta was able to explain two different phenomena, it was assumed to be a universal theory of nature. The concept of quanta that proposed the packaging of energy gave rise to the theory of quantum mechanics. The manner in which we solved the St. Petersburg Paradox is through the use of the concept of the net utility. The same concept of the net utility was used in behavioral economics by \cite{kahneman1979prospect}. In prospect theory, the importance of relative utility is observed additionally to the absolute utility. The concept of net utility is useful in different disciplines, prescriptive economics (e.g., the St. Petersburg Paradox) and descriptive economics (e.g., behavioral economics). Hence, the net utility should be used as a universal concept of understanding economics in the same way as the theory of energy quanta was assumed to be a universal theory when it explained both the black body radiation problem and the photoelectric effect. In classical economics, we define a rational agent as an agent that maximizes its utility. In light of the universal concept of utility, this is the wrong approach to defining a rational agent. We define a rational agent as an agent that maximizes its net utility. With this definition on maximizing the net utility one is able to explain the behavior of people as it is done in the prospect theory. Secondly, one is able to explain the St. Petersburg Paradox.  

\section{Conclusion}
This paper proposed the use of the concept of the net utility to successfully explain the St. Petersburg problem. The concept of net utility or relative utility was experimentally observed by Kahneman and Tvesrky and has had a profound impact in the field of behavioral economics. This paper furthermore concludes that the net utility concept more universally explains value than the nominal utility.

\bibliographystyle{ecca}
\bibliography{muller}
\end{document}